\documentclass{llncs}
\usepackage{amsmath,amssymb}
\usepackage[ruled,vlined,linesnumbered,commentsnumbered]{algorithm2e}
\usepackage{graphicx}
\usepackage{subfig}
\usepackage{epstopdf}
\usepackage{float}
\usepackage{tabularx}
\begin{document}
\title{On Auxiliary Entity Allocation Problem in Multi-layered Interdependent Critical Infrastructures}
\date{\vspace{-5ex}}
\author{\vspace{-5ex}}
\institute{\vspace{-5ex}}
\author{Joydeep Banerjee, Arunabha Sen and Chenyang Zhou}
\institute{School of Computing, Informatics and Decision System Engineering\\
\small Arizona State University, Tempe, Arizona 85287 \\
\small Email: \{joydeep.banerjee, asen, czhou24\}@asu.edu}

\maketitle 
\begin{abstract}
Operation of critical infrastructures are highly interdependent on each other. Such dependencies causes failure in these infrastructures to cascade on an initial failure event. Owing to this vulnerability it is imperative to incorporate efficient strategies for their protection. Modifying dependencies by adding additional dependency implications using entities (termed as \emph{auxiliary entities}) is shown to mitigate this issue to a certain extent. With this finding, in this article we introduce the Auxiliary Entity Allocation problem. The objective is to maximize protection in Power and Communication infrastructures using a budget in number of dependency modifications using the auxiliary entities. The problem is proved to be NP-complete in general case. We provide an optimal solution using Integer Linear program and a heuristic for a restricted case. The efficacy of heuristic with respect to the optimal is judged through experimentation using real world data sets with heuristic deviating $6.75 \%$ from optimal on average. 
\end{abstract}


\section{Introduction}
\label{Intro}
Critical infrastructures like power, communication, transportation networks etc. interact symbiotically to carry out their functionalities. As an example there exists strong mutual interactions between the power and communication network or infrastructure (in this article the term \emph{infrastructure} and \emph{network} are used interchangeably). Entities in the power network like generators, substations, transmission lines etc. relies on control signals carried over by communication network entities like routers, fiberoptic lines etc. Similarly all entities in the communication network relies on power supply from the power network to drive their functionalities. To capture this kind of dependencies the critical infrastructure can be modeled as a multilayered interdependent network. Failure of entities in either infrastructure impacts the operation of its own infrastructure as well as the other infrastructure. Owing to these dependencies the initial failure might result in cascade of failures resulting in disastrous impact. This has been observed in power blackouts which occurred in New York (2003) ~\cite{andersson2005causes} and India (2012) ~\cite{tang2012analysis}.

To study the nature of failure propagation in these interdependent networks it is imperative to model their dependencies as accurately as possible. Recent literature consists of a plethora of these models \cite{Bul10}, \cite{Sha11}, \cite{Gao11}, \cite{Gao11}, \cite{Sha11}, \cite{Ros08}, \cite{Zha05}, \cite{Par13}, \cite{Ngu13}, \cite{Zus11}. However each of these models have their own shortcoming in capturing the complex dependencies that might exist. For example consider a scenario with one power network entity $a_1$ and three communication entities $b_1, b_2, b_3$. The entity $a_1$ is operational provided that both entities $b_1$ \emph{and} $b_2$ are operational \emph{or} if entity $b_3$ is operational (note that the italicized words represent logical operations). None of the above models can accurately model this kind of a dependency. Sen et. al. in \cite{sen2014identification} proposed a model that uses boolean logic to capture these interdependencies. This model is referred to as the Implicative Interdependency model (IIM)). To express the dependency of an entity on other entities it uses implications which are disjunction(s) and conjunction(s) of logical terms (denoting entities of the network). With respect to the example considered above the dependency implication for the entity $a_1$ can be represented as $a_1 \leftarrow b_1 b_2 + b_3$. The boolean implication depicting the dependency is termed as \emph{Inter-Dependency Relation}. Our approach in designing solutions and analyzing the problem addressed in this paper is based on the IIM model.

We restrict our attention to an interdependent power and communication network in this paper. However the solutions can be extended to any two interdependent networks. As discussed earlier initial failure of a certain entity set in power and communication network may trigger cascading failure resulting in loss of a large number of entities. Authors in \cite{BanHardening15} proposed the \emph{Entity Hardening} problem to increase the reliability of these interdependent systems. They assumed that an entity when hardened would be resistant to both initial and cascading failure. Given a set of entities that failed initially (that is at time $t=0$) the problem was to find a minimal set of entities which when hardened would prevent failure of at least a predefined number of entities. On situations where entity hardening is not possible alternative strategies needs to be employed to increase the system reliability. Adding additional dependencies for entities in interdependent infrastructure can be beneficial in this regard. We elaborate this with the help of an example. Consider the dependency rule $a_1 \leftarrow b_1 b_2 + b_3$. With this dependency entity $a_1$ would fail if entities $(b_1,b_3)$ or $(b_2,b_3)$ fails. Now consider an entity $b_4$ is added as a disjunction to the IDR (with the new dependency being $a_1 \leftarrow b_1 b_2 + b_3 + b_4$). For entity $a_1$ to fail, either $(b_1, b_3, b_4)$ or $(b_2,b_3,b_4)$ should fail. This increases the reliability compared to the previous dependency for $a_1$. Hence adding additional dependency can be employed as a strategy when entity hardening is not possible. Any entity added to modify a dependency relation is termed as an \emph{auxiliary entity}. However due to system, cost and feasibility constraints the number of such modifications are restricted. Hence  when the number of IDR modifications are restricted one has to find which IDRs to modify and with what entities so that the impact of failure is minimized. We term this problem as the \emph{Auxiliary Entity Allocation Problem}. It is to be noted that in both \emph{Entity Hardening Problem} and Auxiliary Entity Allocation Problem the IDRs of the interdependent system are changed but these changes are carried out differently. 

The rest of the paper is organized as follows. A brief explanation of the IIM model with formal problem definition is provided in Section \ref{ProbForm}. The computational complexity of the problem and proposed solutions are discussed in Sections \ref{CompAna} and \ref{Sol} respectively. We discuss the experimental results in Section \ref{ExpRes} and conclude the paper in Section \ref{Conc}

\section{Problem Formulation using the Implicative Interdependency Model }
\label{ProbForm}
We briefly describe the IIM model introduced in \cite{sen2014identification}. Two sets $A$ and $B$ represent entities in power and communication network. The dependencies between these set of entities are captured using a set of interdependency denoted as $\mathcal{F}(A,B)$. Each function in the set $\mathcal{F}(A,B)$ is termed as an \emph{Inter-Dependency Relation} (IDR). We describe an interdependent network which composes of the entity sets $A$ and $B$ and the interdependency relations $\mathcal{F}(A,B)$ and denote it by $\mathcal{I}(A,B,\mathcal{F}(A,B))$. Through an example we explain this model further. Consider an interdependent network $\mathcal{I}(A,B,\mathcal{F}(A,B))$ with $A = \{a_1,a_2,a_3,a_4,a_5\}$ and $B=\{b_1,b_2,b_3\}$. The set of IDRs ($\mathcal{F}(A,B)$) for the interdependent network are provided in Table \ref{tbl:example1idr}. Consider the IDR $a_3 \leftarrow b_2 + b_1 b_3$ in Table \ref{tbl:example1idr}. It implies that the entity $a_1$ is operational if entity $b_2$ \emph{or} entity $b_1$ \emph{and} $b_3$ are operational. As evident, the IDRs are essentially disjunction(s) of entity (entities) in conjunctive form. We refer to each conjunctive term, e.g. $b_1 b_3$, as \emph{minterm}. The example considers dependencies where an entity in network A(B) is dependent on entities in network B(A) i.e. inter-network dependency. However this model can capture intra-network dependencies as well.

\begin{table}
	\parbox{.45\linewidth}{
		\centering
		\begin{tabular}{|l||l|}  \hline
			{\bf Power Network} & {\bf Comm. Network} \\ \hline
			$a_1 \leftarrow b_1 + b_2$ & $b_1\leftarrow a_2 $ \\ \hline
			$a_2 \leftarrow b_1b_2$ & $b_2 \leftarrow  a_2$ \\ \hline
			$a_3 \leftarrow b_2 + b_1 b_3$ & $b_3 \leftarrow a_4$ \\ \hline
			$a_4 \leftarrow b_3$ & 	$--$\\ \hline
			$a_5$ & 	$--$\\ \hline
		\end{tabular}
		\vspace{0.08in}
		\caption{IDRs for the constructed example}
		\protect\label{tbl:example1idr}
	}
	\hfill
	\parbox{.45\linewidth}{
		\centering
		\begin{tabular}{|c|c|c|c|c|c|c|c|c|}  \hline
			\multicolumn{1}{|c|}{Entities} & \multicolumn{8}{c|}{Time Steps ($t$)}\\
			\cline{2-8} & $0$ & $1$ & $2$ & $3$ & $4$ & $5$ & $6$ & $7$ \\\hline \hline
			$a_1$ & $0$ & $0$ & $0$ & $1$ & $1$ & $1$ & $1$ & $1$\\ \hline
			$a_2$ & $0$ & $1$ & $1$ & $1$ & $1$ & $1$ & $1$ & $1$ \\ \hline
			$a_3$ & $0$ & $1$ & $1$ & $1$ & $1$ & $1$ & $1$ & $1$ \\ \hline
			$a_4$ & $0$ & $1$ & $1$ & $1$ & $1$ & $1$ & $1$ & $1$ \\ \hline
			$a_5$ & $0$ & $0$ & $0$ & $0$ & $0$ & $0$ & $0$ & $0$ \\ \hline
			$b_1$ & $0$ & $0$ & $1$ & $1$ & $1$ & $1$ & $1$ & $1$ \\ \hline
			$b_2$ & $1$ & $1$ & $1$ & $1$ & $1$ & $1$ & $1$ & $1$ \\ \hline
			$b_3$ & $1$ & $1$ & $1$ & $1$ & $1$ & $1$ & $1$ & $1$ \\ \hline
		\end{tabular}
		\vspace{0.08in}
		\caption{Cascade propagation when entities $\{b_2,b_3\}$ fail initially. $0$ denotes the entity is operational and $1$ non-operational}
		\protect\label{tbl:example1cascade}
	}
\vspace{-0.2in}
\end{table}

The cascading procedure is described with respect to the interdependent network captured by IDRs in Table \ref{tbl:example1idr}. The cascade proceeds in unit time steps (denoted by $t$). Consider two entities $b_2$ and $b_3$ are attacked and made non operational by an adversary at time step $t=0$ (\emph{initial failure}). Owing to the IDRs $a_2 \leftarrow b_1 b_2$, $a_3 \leftarrow b_2 + b_1 b_3$ and $a_4 \leftarrow b_3$ the entities $a_2, a_3, a_4$ becomes non operational at $t=1$. Subsequently entities $b_1$ ($b_1 \leftarrow a_2$) and $a_1$ ($a_1 \leftarrow b_1 + b_2$) seize to operate at time step $t=2$ and $t=3$ respectively. The failure of entities after $t=0$ is termed as \emph{induced failure}. The cascade is represented in Table \ref{tbl:example1cascade}. It is to be noted that the maximum number of time steps in the cascade for any interdependent network is $|A| + |B| - 1$ (assuming initial time step as $t=0$). Hence in Table \ref{tbl:example1cascade} with number of entities being $8$ the state (operational or non-operational) of all entities are shown till $t=7$. Construction of these IDRs is a major challenge of this model. Possible strategies are (i) deep investigation of physical properties and flows in the interdependent network and (ii) consultation with domain experts. The methodology to construct these IDRs is ongoing and is expected to be addressed in future. The problem in this article assumes that the IDRs are already constructed for a given interdependent network. 

Authors in \cite{sen2014identification} introduced the $\mathcal{K}$ most vulnerable entities problem. The problem used the IIM model to find a set of $\mathcal{K}$ (for a given integer $|\mathcal{K}|$) entities in an interdependent network $\mathcal{I}(A,B,\mathcal{F}(A,B))$ whose failure at time $t=0$ (\emph{initial failure}) would result in failure of the largest number of entities due to \emph{induced failure}. For the example provided in this section consider the case where $|\mathcal{K}| = 2$. Failing entities $b_2$ and $b_3$ at $t=0$ make entities $\{a_1,a_2,a_3,a_4,b_1,b_2,b_3\}$ not operational by $t=3$. Hence $\mathcal{K} = \{b_2, b_3\}$ are one of the $2$ most vulnerable entities in the interdependent network (it is possible to have multiple $\mathcal{K}$ most vulnerable entities in an interdependent network). The set of entities failed when $\mathcal{K}$ most vulnerable entities fail initially is denoted by $A' \cup B'$ with $A' \subseteq A$ and $B' \subseteq B$. Here $A' = \{a_1,a_2,a_3,a_4\}$ and $B'=\{b_1,b_2,b_3\}$.

For a given $\mathcal{K}$ most vulnerable entities of an interdependent network, the system reliability can be increased (i.e. entities can be protected from failure) by \emph{Entity Hardening} \cite{BanHardening15}. On scenarios where entity hardening is not possible it is imperative to take alternative strategies. The number of entities failing due to \emph{induced failure} can be reduced by modifying the IDRs. One way of modifying an IDR is adding an entity as a new minterm. For example, consider the interdependent network with IDRs given by Table \ref{tbl:example1idr} and $b_2$ and $b_3$ being the $2$ (when $\mathcal{K}$ = 2) most vulnerable entities (as discussed above). Let the IDR $b_1 \leftarrow a_2$ be modified as $b_1 \leftarrow a_2+ a_5$. Hence the new interdependent network is represented as $\mathcal{I}(A,B,\mathcal{F}'(A,B))$ with the same set of IDRs as that in Table \ref{tbl:example1idr} except for IDR $b_1 \leftarrow a_2+ a_5$ as the sole modification. The entity $a_1$ introduced is termed as an \emph{auxiliary entity}. It follows that after the modification, failure of entities $b_2$ and $b_3$ at time $t=0$ would trigger failure of entities $a_2,a_3$ and $a_4$ only. Thus before modification the failure set would have been $\{a_1,a_2,a_3,a_4,b_1,b_2,b_3\} $ and after the modification it would be $\{a_2,a_3,a_4,b_2,b_3\}$. Thus the modification would lead to a fewer number of failures.

We make the following assumptions while modifying an IDR --- 
\begin{itemize}
\item It is possible to add an auxiliary entity as conjunction to a minterm. However it is intuitive that this would have no impact in decreasing the number of entities failed due to \emph{induced failure}. Hence we modify an IDR by adding only one auxiliary entity as a disjunction to a minterm
\item An auxiliary entity does not have the capacity to make an entity operational which fails due to \emph{initial failure}. So to prune the search set for obtaining a solution we discard entities in $(A' \cup B')$ as possible auxiliary entities.
\item If an IDR $D$ is modified then it is done by adding only one entity not in $A' \cup B' \cup E_D$ where $E_D$ is a set consisting of all entities (both on left and right side of the equation) in $D$. For any IDR $D \in \mathcal{F}(A,B)$ we denote $AUX = (A \cup B) / (A' \cup B' \cup E_D)$ as the set of auxiliary entities that can be used to modify $D$. 
\end{itemize}

We quantify the number of modifications done as the number of IDRs to which minterms are added as a disjunction. It should also be noted than an attacker only have information about the initial interdependent network $\mathcal{I}(A,B,\mathcal{F}(A,B))$. Hence with a budget of $|\mathcal{K}|$ it attacks and kills the $\mathcal{K}$ most vulnerable entities to maximize the number of entities killed due to induced failure. Any modification in the IDR is assumed to be hidden from the attacker. 

With these definitions the Auxiliary Entity Allocation Problem (AEAP) is defined as follows. Let $\mathcal{K}$ be the most vulnerable entities (already provided as input) of an interdependent network $\mathcal{I}(A,B,\mathcal{F}(A,B))$. With a budget $S$ in number of modifications, the task is to find which are the $S$ IDRs that are to be modified and which entity should be used to perform this modification such that number of entities failing due to \emph{induced failure} is minimized. A more formal description given below. 
\\ \\
\textbf{The Auxiliary Entity Allocation Problem (AEAP)}\\
\textbf{Instance} --- An interdependent network $\mathcal{I}(A,B,\mathcal{F}(A,B))$, $\mathcal{K}$ most vulnerable entities for a given integer $|\cal{K}|$ and two positive integers $S$ and $P_f$.\\ 
\textbf{Decision Version} --- Does there exist $S$ IDR auxiliary entity tuple $(D,x_i)$ such that when each IDRs $D\in \mathcal{F}(A,B)$ is modified by adding auxiliary entity $x_i \in AUX$ as a disjunction it would protect at least $P_f$ entities from \emph{induced failure} with $\mathcal{K}$ vulnerable entities failing initially. 

\section{Computational Complexity Analysis}
\label{CompAna}
In this section we analyze the computational complexity of the AEAP problem. The computational complexity of the problem depends on nature of the IDRs. The problem is first solved by restricting the IDRs to have one minterm of size $1$. For this special case a polynomial time algorithm exists for the problem. With IDRs in general form the problem is proved to be NP-complete.

\subsection{Special Case: Problem Instance with One Minterm of Size One}
The special case consist of IDRs which have a single minterm of size $1$ and each entity appearing exactly once on the right hand side of the IDR. With entities $a_{i}$'s and $b_{j}$'s belonging to network $A(B)$ and $B(A)$ respectively, the IDRs can be represented as $a_i \leftarrow b_j$. The AEAP problem can be solved in polynomial time for this case. We first define \emph{Auxiliary Entity Protection Set} and use it to provide a polynomial time heuristic in Algorithm \ref{heu:heu1}. The proof of optimality is not included due to space constraint. 

\begin{definition}
	\textbf{Auxiliary Entity Protection Set:} With a given set of $\mathcal{K}$ most vulnerable entities failing initially the Auxiliary Entity Protection Set is defined as the number of entities protected from \emph{induced failure} when an auxiliary entity $x_i$ is added as a disjunction to an IDR $D \in \mathcal{F}(A,B)$. It is denoted as $AP(D,x_i|\mathcal{K})$.  
	\vspace{-8mm}
\end{definition} 

\begin{algorithm}
	\small
	\KwData{An interdependent network $\mathcal{I}(A,B,\mathcal{F}(A,B))$ and set of $\mathcal{K}$ vulnerable entities}
	\KwResult{A set $D_{sol}$ consisting of IDR auxiliary entity doubles $(D,x_i)$ (with $|D_{sol}| = S$ and $P_f$ (denoting the entities protected from induced failure)}
	\Begin{			
		For each IDR $D \in \mathcal{F}(A,B)$ and each entity $x_i \in AUX$ (where $AUX = A \cup B / (A' \cup B' \cup E_D)$ as discussed in the previous section) compute the \emph{Auxiliary Entity Protection Set} $AP( D,x_i|\mathcal{K})$ \;
		Initialize $D_{sol}=\emptyset$ and $P_f = \emptyset$\;
		\While  {$S \ne 0$}{
			Choose the Auxiliary Entity Protection Set with highest $AP(x_i, D|\mathcal{K})$. In case of tie break arbitrarily. Let $D_{cur}$ be the corresponding IDR and $x_{cur}$ the auxiliary entity\;
			Update $D_{sol} = D_{sol} \cup (D_{cur},x_{cur})$ and add auxiliary entity $x_{cur}$ as a disjunction to the IDR $D_{cur}$ \;
			Update $P_f = P_f \cup AP(D_{cur}, x_{cur}|\mathcal{K})$\;
			\For{$\forall \text{ IDR } D' \in \mathcal{F}(A,B) \text{ and } x_i \in AUX \text{ of } D'$}{
				Update $AP(D',x_i| \mathcal{K}) = AP(D', x_i|\mathcal{K}) \backslash AP(D_{cur}, x_{cur}|\mathcal{K})$\;
			}
			$S \leftarrow S-1$\;
		}
	\Return{$D_{sol}$ and $P_f$} \;
}	
\caption{Algorithm solving AEAP problem for IDRs with minterms of size $1$}
\label{heu:heu1}
\end{algorithm}

\subsection{General Case: Problem Instance with an Arbitrary Number of Minterms of Arbitrary Size}
The IDRs in general are composed of disjunctions of entities in conjunctive form i.e. arbitrary number of minterms of arbitrary size. This case can be represented as  $a_{i} \leftarrow \sum^{p}_{k_1=1} \prod^{j_{k_1}}_{j=1}b_{j}$ where entities $a_{i}$ and $b_{j}'s$ belong to network $A(B)$ and $B(A)$ respectively. The given example has $p$ minterms each of size $j_{k_1}$. In Theorem \ref{th:thm4} we prove that the decision version of the AEAP problem for general case is NP complete. 

\begin{theorem}
	\label{th:thm4}
	The decision version of the AEAP problem for Case IV is NP-complete.
\end{theorem}

\begin{proof}
	The hardness is proved by a reduction from Set Cover problem. An instance of a set cover problem consists of a universe $U =\{x_1,x_2,...,x_n\}$ of elements and set of subsets $\mathcal{S} = \{S_1,S_2,...,S_m\}$ where each element $S_i \in \mathcal{S}$ is a subset of $U$. Given an integer $X$ the set cover problem finds whether there are $\le X$ elements in $\mathcal{S}$ whose union is equal to $U$. From an instance of the set cover problem we create an instance of the AEAP problem. For each subset $S_i$ we create an entity $b_i$ and add it to set $B$. For each element $x_j$ in $U$ we add an entity $a_j$ to a set $A_1$. We have a set $A_2$ of entities where $|A_2| = |B|$. Let $A_2 = \{ a_{21}, a_{22}, ...,a_{2|B|}\}$ where there is an association between entity $b_j$ and $a_{2j}$. Additionally we have a set of entities $A_3$ with $|A_3| = X$ which does not have any dependency relation of its own. The set $A$ is comprised of $A_1 \cup A_2 \cup A_3$. The IDRs are created as follows. For an element $x_i$ that appears in subsets $S_x, S_y, S_z$, an IDR $a_i \leftarrow b_x + b_y + b_z$ is created. For each entity $b_j \in B$ an IDR $b_j \leftarrow a_{2j}$ is added to $\mathcal{F}(A,B)$. The cardinality of $\mathcal{K}$ most vulnerable node is set to $|A_2|$ and it directly follows that $\mathcal{K} = A_2$ comprises the set of most vulnerable entities. The value of $S$ (number of IDR modifications) is set to $X$ and $P_f$ is set to $S + |A_1|$.
	  
	Let there exist a solution to the set cover problem. Then there exist at least $X$ subsets whose union covers the set $U$. For each subset $S_k$ which is in the solution of the set cover problem we choose the corresponding entity $b_k$. Let $B'$ be all such entities. We arbitrarily choose and add an entity from $A_3$ to each IDR $b_k \leftarrow a_{2k}$ with $b_k \in B'$ to form $S = X$ distinct IDR auxiliary entity doubles. As $A_3$ type entities does not have any dependency relation thus all the entities in $B$ that correspond to the subsets in the solution will be protected from failure. Additionally protecting these $B$ type entities would ensure all entities in $A_1$ does not fail as well (as there exists at least one $B$ type entity in the IDR of $A_1$ type entities which is operational). Hence a total of $X + |A_1|$ are protected from failure.    
	  
	Similarly let there exist a solution to the AEAP problem. It can be checked easily that no entities in $B \cup A_1 \cup A_2$ has the ability to protect additional entities using IDR modification. Hence set $A_3$ can only be used as auxiliary entities. An entity from $A_3$ for the created instance can be added to an IDR of $A_1$ type entity or $B$ type entity. In the former strategy only one entity is protected from failure whereas two entities are operational when we add auxiliary entity to IDRs of $B$ type entities. Hence all the auxiliary entities are added to the $B$ type IDRs with a final protection of $X + |A_1|$ entities. For each IDR of the $B$ type entity to which the auxiliary entity is added, the corresponding subset in $\mathcal{S}$ is chosen. The union of these subsets would result in $U$ as the solution of the AEAP problem protects the failure of all $A_1$ type entities. Hence solving the set cover problem and proving the hardness stated in Theorem \ref{th:thm4}.
\end{proof}

\section{Solutions to the AEAP Problem}
\label{Sol}
We consider the following restricted case where there exists at least $S$ entities in the interdependent network which does not belong to any of the failing entities. This comprise the set of auxiliary entities that can be used. It is also imperative to use such set as auxiliary entities because they never fail from induced or initial failure when the $\mathcal{K}$ most vulnerable entities fail initially. The problem still remains to be NP compete for this case as in Theorem \ref{th:thm4} the set of entities $A_3$ belong to such class of auxiliary entities. With these definition of the special case let $\mathcal{A}$ denote a set of such auxiliary entities which can be used for IDR modifications with $\mathcal{A} \subset A \cup B / (A'' \cup B'')$ (where $A'' \cup B''$ are the entities that fails due to failing entities $A' \ cup B'$ initially). Hence we loose the notion of IDR auxiliary entity doubles in the solution as any auxiliary entity from set $\mathcal{A}$ would produce the same protection effect. Let $\mathcal{A}$ denote all such entities that can be used as auxiliary entities as defined above. We additionally assume that $|\mathcal{A}| \ge S$, i.e., there are enough auxiliary entities to suffice the AEAP budget $S$. So in both the solutions we only consider the IDRs that needs to be modified and disregard which auxiliary entity is used for this modification. We first propose an Integer Linear Program (ILP) to obtain the optimal solution in this setting. We later provide a polynomial heuristic solution to the problem. The performance of heuristic with respect to the ILP is compared in the section to follow.

\subsection{Optimal solution to AEAP problem}
We first define the variables used in formulating the ILP. Two set of variables $G = \{g_1,g_2,...,g_c\}$ and $H = \{h_1,h_2,...,h_d\}$ (with $c = |A|$ and $d = |B|$) are used to maintain the solution of $\mathcal{K}$ most vulnerable entities. Any variable $g_i \in G$ ($h_j \in H$) is equal to $1$ if $a_i \in A$ ($b_j \in B$) belongs to $\mathcal{K}$ and is $0$ otherwise. For each entity $a_i$ and $b_j$ a set of variables $x_{id}$ and $y_{jd}$ are introduced with $0 \le d \le |A| + |B| -1$. $x_{id}$ ($y_{id}$) is set to $1$ if the entity $a_i$ ($b_j$) is non operational at time step $d$ and is $0$ otherwise. Let $P$ denote the total number of IDRs in the interdependent network and assume each IDR has a unique label between numbers from $1$ to $P$. A set of variables $M = \{m_1, m_2, ...,m_{P}\}$ are introduced. The value of $m_i$ is set to $1$ if an auxiliary node is added as a disjunction to the IDR labeled $i$ and $0$ otherwise. With these definitions we define the objective function and the set of constraints in the ILP. 

\noindent
\begin{equation}\label{eqn:ilpobj1}
min \Big(\overset{|A|}{\underset{i=1}{\sum}}x_{i(|A| +|B| - 1)}+\overset{|B|}{\underset{j=1}{\sum}}y_{j(|A| + |B| -1)} \Big)
\end{equation}

\noindent
The objective function defined in \ref{eqn:ilpobj1} tries to minimize the number of entities having value $1$ at the end of the cascade i.e. time step $|A| + |B| - 1$. Explicitly this objective minimizes the number of entities failed due to induced failure. The constraints that are imposed on these objective to capture the definition of AEAP are listed below ---

\noindent
{\em Constraint Set 1:} $x_{i0} \ge g_{i}$ and $y_{j0} \ge h_{j}$. This imposes the criteria that if entity $a_i$ ($b_j$) belongs to the $\mathcal{K}$ most vulnerable entity set then the corresponding variable $x_{i0}$ ($y_{j0}$) is set to $1$ capturing the \textit{initial failure}.\\
\noindent \\
{\em Constraint Set 2}: $x_{id} \geq x_{i(d-1)}, \forall d, 1 \leq d \leq |A|+|B|-1$, and  $y_{id} \geq y_{i(d-1)}, \forall d, 1 \leq d \leq |A|+|B|-1$,. This ensures that the variable corresponding to an entity which fails at time step $t$ would have value $1$ for all $d \ge t$.\\
\noindent \\
{\em Constraint Set 3}: We use the theory developed in \cite{sen2014identification} to generate constraints to represent the cascade through the set of IDRs. To describe this consider an IDR ${a_i} \leftarrow {b_j}{b_p}{b_l} + {b_m}{b_n} + {b_q}$ in the interdependent network. Assuming the IDR is labeled $v$ it is reformulated as ${a_i} \leftarrow {b_j}{b_p}{b_l} + {b_m}{b_n} + {b_q} + m_v$ with $m_v \in M$. This is done for all IDRs. The constraint formulation is described in the following steps.

\vspace{0.02in}
\noindent
{\em Step 1:} All minterms of size greater than $1$ are replaced with a single virtual entity. In this example we introduce two virtual entities $C_1$ and $C_2$ ($C_1, C_2 \notin A \cup B$) capturing the IDRs $C_1 \leftarrow {b_j}{b_p}{b_l}$ and $C_2 \leftarrow {b_m}{b_n}$. The IDR in the example can be then transformed as ${a_i} \leftarrow C_1 + C_2 + b_q + m_v$. For any such virtual entity $C_k$ a set of variables $c_{kd}$ are added with $c_{kd} = 1$ if $C_k$ is alive at time step $d$ and $0$ otherwise. Hence all the IDRs are represented as disjunction(s) of single  entities. Similarly all virtual entities have IDRs which are conjunction of single entities. 

\vspace{0.02in}
\noindent
{\em Step 2:} For a given virtual entity $C_k$ and all entities having a single midterm of arbitrary size, we add constraints to capture the cascade propagation. Let $N$ denote the number of entities in the IDR of $C_k$. The constraints added is described through the example stated above. The variable $c_1$ with IDR $C_1 \leftarrow {b_j}{b_p}{b_l}$, constraints $c_{1d} \geq \frac{y_{j(d-1)} + y_{p(d-1)}+y_{l(d-1)}}{N}$ and $c_{1d} \le y_{j(d-1)} + y_{p(d-1)}+y_{l(d-1)} \forall d, 1 \leq d \leq m+n-1$ are added (with $N=3$ in this case). This ensures that if any entity in the conjunction fails the corresponding virtual entity fails as well. 

\vspace{0.02in}
\noindent
{\em Step 3:} In the transformed IDRs described in step 1 let $n$ denote the number of entities in disjunction for any given IDR (without modification). In the given example with IDR ${a_i} \leftarrow C_1 + C_2 + b_q + m_v$, constraints of form $x_{id} \geq c_{1(d-1)} + c_{2(d-1)} + y_{q(d-1)} + m_v - (n - 1)$ and $x_{id} \leq \frac{c_{1(d-1)} + c_{2(d-1)} + y_{q(d-1)} + m_v}{n} \forall d, 1 \leq d \leq m+n-1$ are added. This ensures that the entity $a_i$ will fail only if all the entities in disjunction become non operational. \\ \\
{\em Constraint Set 4:} To ensure that only $S$ auxiliary entities are added as disjunction to the IDRs constraint $\sum_{v=1}^{P} m_v = S$ is introduced.

\subsection{Heuristic solution to the AEAP problem}
In this section we provide a polynomial heuristic solution to the AEAP problem. We first redenote \emph{Auxiliary Entity Protection Set} as $AP(D|\mathcal{K})$ as it is immaterial which entity is added as an auxiliary entity since no auxiliary entity can fail due to any kind of failure. Along with the definition of \emph{Auxiliary Entity Protection Set}, we define \textit{Auxiliary Cumulative Fractional Minterm Hit Value} (ACFMHV) for designing the the heuristic. We first define \textit{Auxiliary Fractional Minterm Hit Value} (AFMHV) in Definition \ref{FM} which is used in defining ACFMHV (in Definition \ref{CFM}).

\begin{definition}
	\label{FM}
	The Auxiliary Fractional Minterm Hit Value of an IDR $D \in \mathcal{F}(A,B)$ is denoted by $AFMHV(D|\mathcal{K})$. It is calculated as $AFMHV(D|\mathcal{K}) = \sum_{i = 1}^{m} \frac{1}{|s_i|}$. Let $x_j$ denote the entity in the right hand side of the IDR $D$ and $m$ denotes all the minterms in which the entity $x_j$ appears over all IDRs. The parameter $s_i$ denotes $i^{th}$ such minterm with $|s_i|$ being its size. If an auxiliary entity is placed at $D$ then the value computed above provides an estimate implicit impact on protection of other non operational entities.  
\end{definition} 

\begin{definition}
	\label{CFM}
	The Auxiliary Cumulative Fractional Minterm Hit Value of an IDR $D \in \mathcal{F}(A,B)$ is denoted by $ACFMHV(D)$. It is computed as $ACFMHV(D) = \sum_{\forall x_i \in AP(D|\mathcal{K})} AFMHV(D_{x_i}|\mathcal{K})$ where $D_{x_i}$ is the IDR for entity $x_i \in AP(D|\mathcal{K})$. The impact produced by the protected entities when IDR $D$ is allocated with an auxiliary entity over set $A \cup B$ is implicitly provided by this definition. 
\end{definition} 

\begin{algorithm}[ht!]
	\small
	\KwData{An interdependent network $\mathcal{I}(A,B,\mathcal{F}(A,B))$, set of $\mathcal{K}$ vulnerable entities, set $\mathcal{A}$ of auxiliary entities and budget $S$}. 
	\KwResult{A set $D_{sol}$ consisting of IDRs (with $|D_{sol}| = S$ to each of which an auxiliary entity is added as a disjunction and $P_f$ (denoting the entities protected from induced failure)}
	\Begin{			
		Initialize $D_{sol}=\emptyset$ and $P_f = \emptyset$\;
		\While{$S \ne 0$}{
			For each IDR $D \in \mathcal{F}(A,B)$ compute the \emph{Auxiliary Node Protection Set} $AP(D|\mathcal{K})$ \;
			\If{There exists multiple IDRs having same value of highest cardinality of the set $AP(D|\mathcal{K})$}{
				For each IDR $D \in \mathcal{F}(A,B)$ compute the \emph{Auxiliary Cumulative Fractional Minterm Hit Value} $ACFMHV(D)$ \;
				Let $D_p$ be an IDR having highest $ACFMHV(D_p)$ among all $D_i$'s in the set of IDRs having highest cardinality of the set $AP(D_i|\mathcal{K})$\;
				If there is a tie choose arbitrarily\;
				Update $D_{sol} = D_{sol} \cup D_p$ and add an auxiliary entity from $\mathcal{A}$ as a disjunction to the IDR $D_p$\;
				Update $P_f = P_f \cup AP(D_{p})$\;
				Update $\mathcal{A}$ by removing the auxiliary entity added \;
				$S \leftarrow S - 1$\;
			}
			\Else{
				Let $D_p$ be an IDR having highest cardinality of the set $D \in \mathcal{F}(A,B)$\;
				Update $D_{sol} = D_{sol} \cup D_p$ and add an auxiliary entity from $\mathcal{A}$ as a disjunction to the IDR $D_p$\;
				Update $P_f = P_f \cup AP(D_{p}|\mathcal{K})$\;
				Update $\mathcal{A}$ by removing the auxilary entity added \;
				$S \leftarrow S - 1$\;
			}	
			Prune the interdependent network $\mathcal{I}(A,B,\mathcal{F}(A,B)$ by removing the IDRs for entities in $AP(D_p| \mathcal{K})$ and removing the same set of entities from $A \cup B$ \;
		}
		\Return{$D_{sol}$ and $P_f$} \;
	}		
	\caption{Heuristic solution to the AEAP problem}
	\label{alg:alg2}
\end{algorithm} 
\vspace{-0.10in}

\begin{figure*}[!htb]
	\centering
	\begin{center}
		\subfloat[][Region 1]{\includegraphics[width=4.5cm]{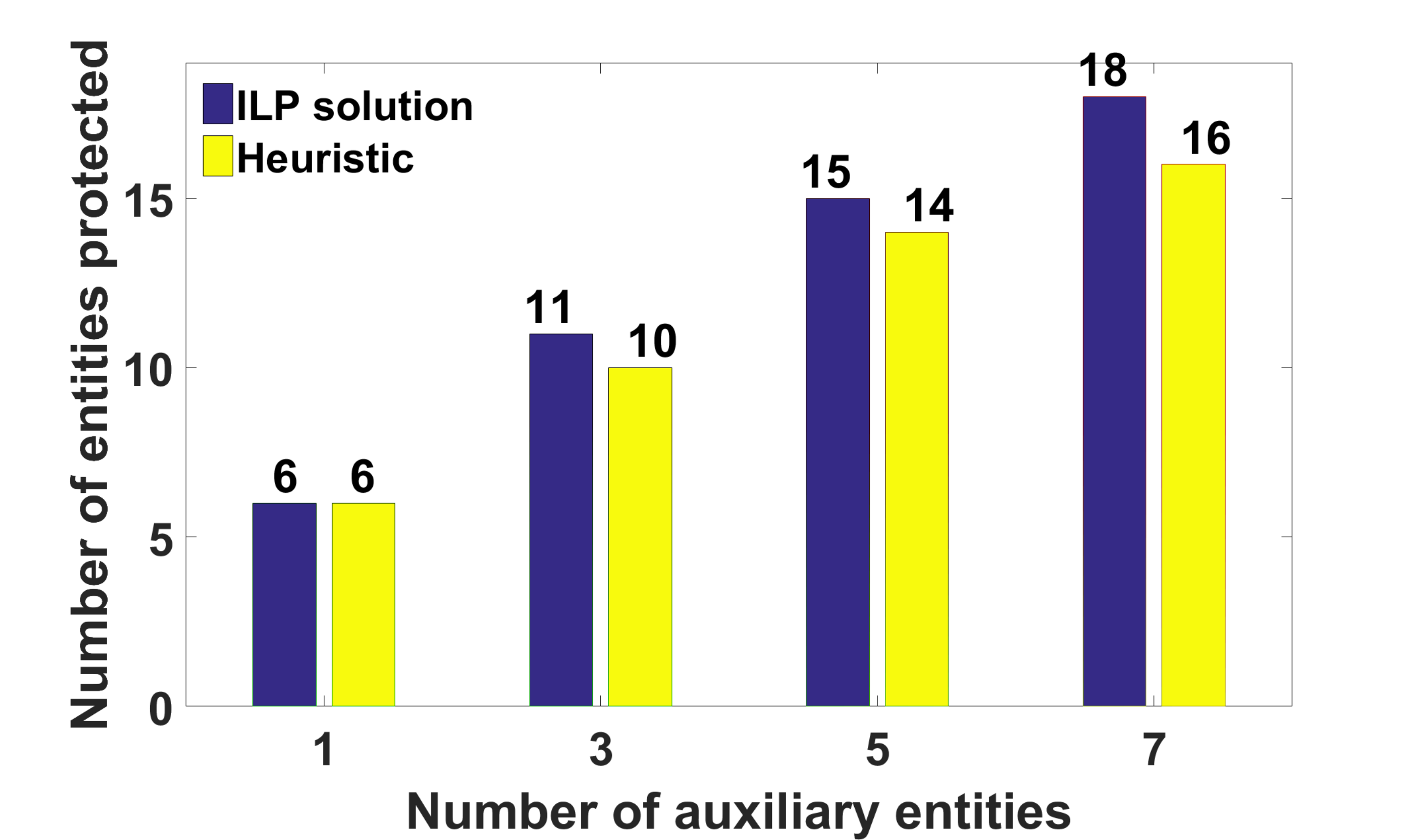}\label{fig:plot1}}
		\subfloat[][Region 2]{\includegraphics[width=4.5cm]{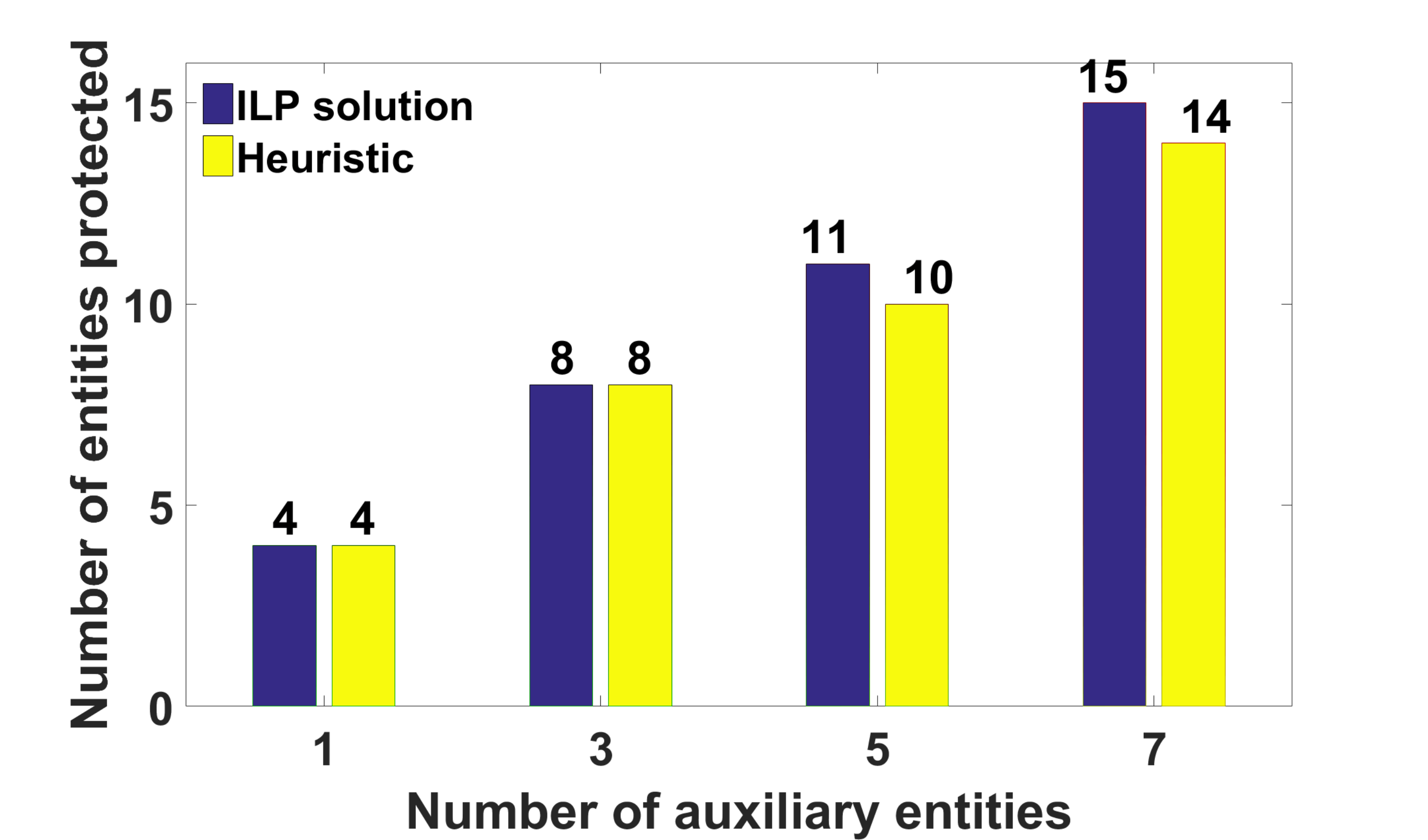}\label{fig:plot2}}\\
		\subfloat[][Region 3]{\includegraphics[width=4.5cm]{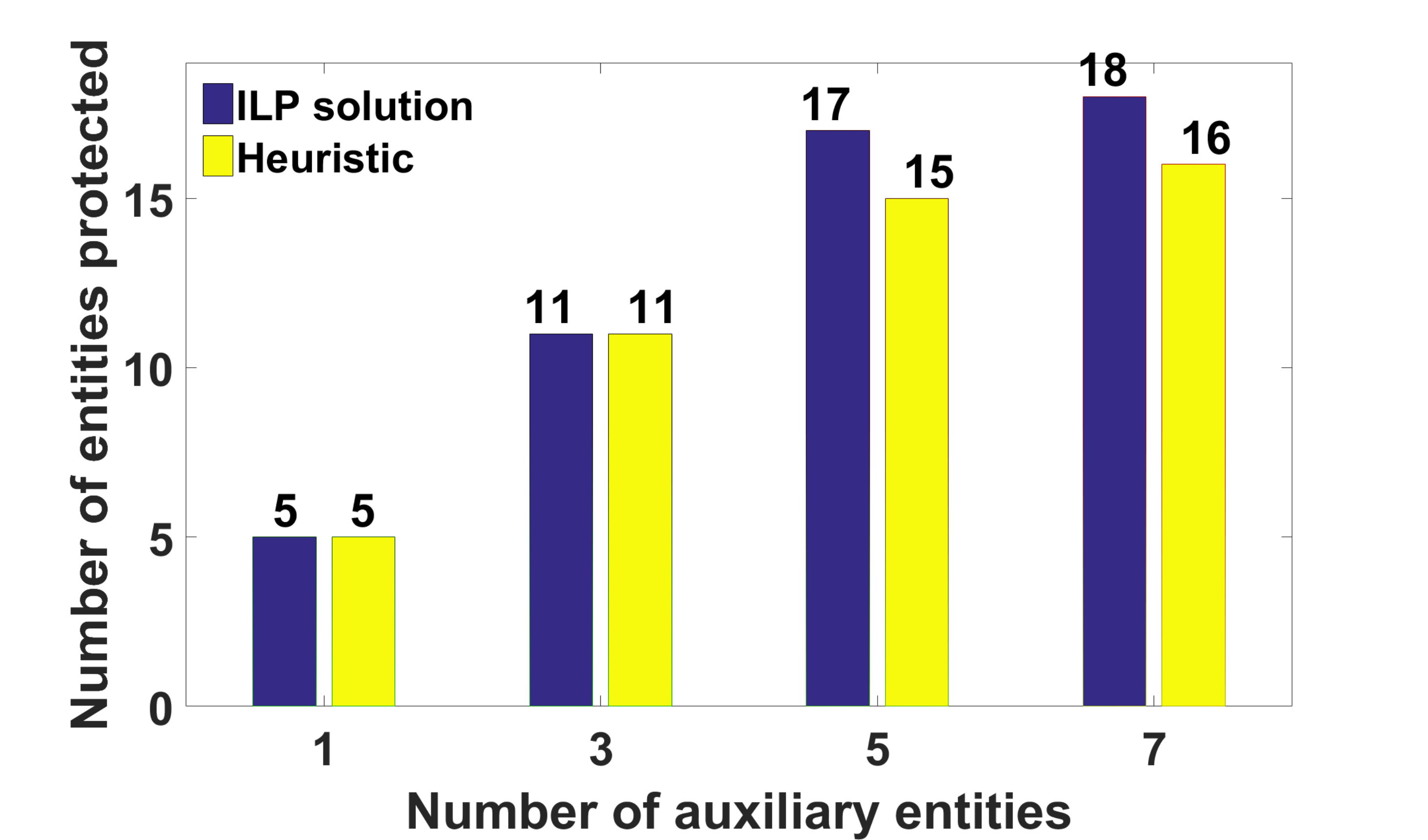}\label{fig:plot3}}
		\subfloat[][Region 4]{\includegraphics[width=4.5cm]{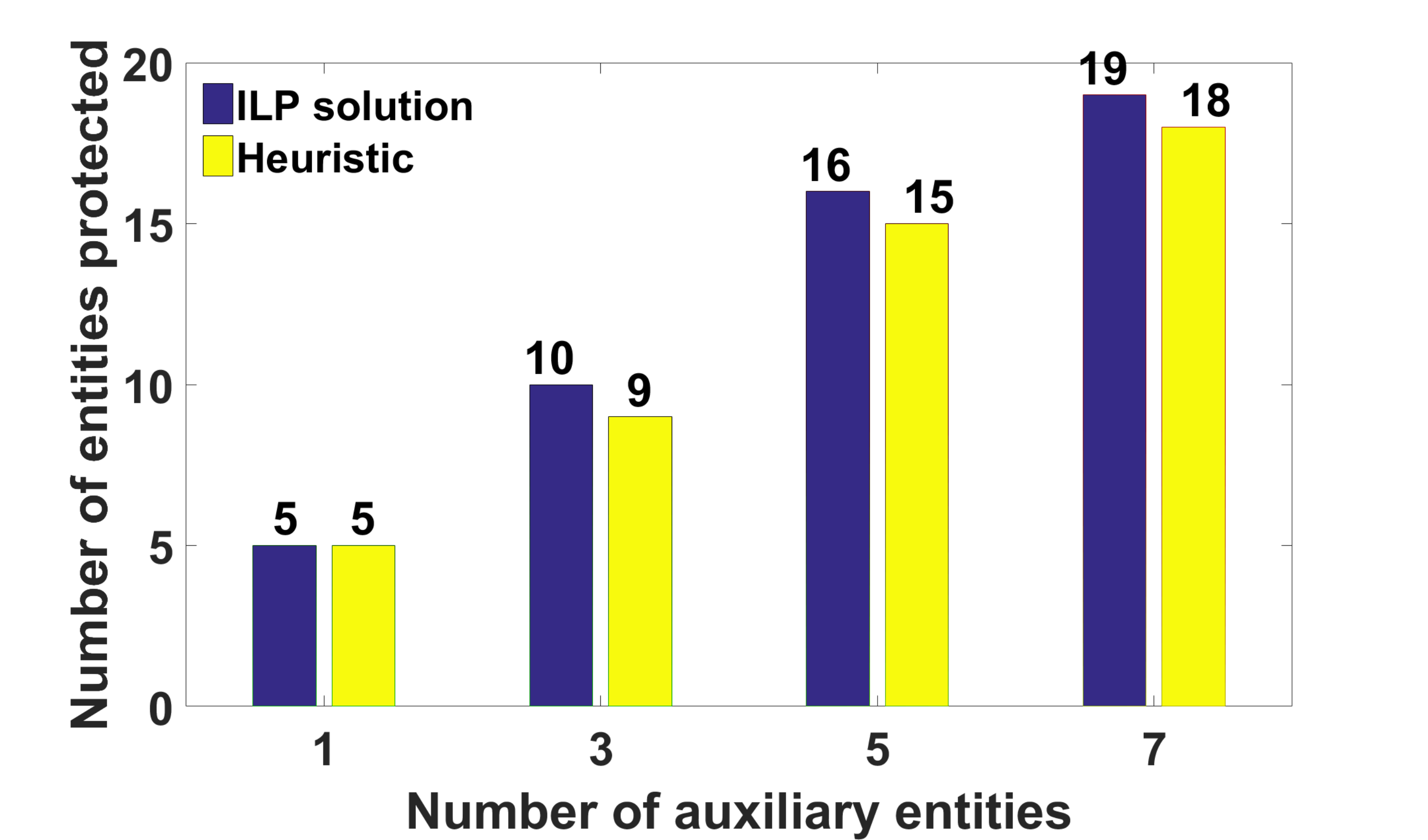}\label{fig:plot4}}
	\end{center}
	\caption{Comparison of the number of entities protected in optimal solution (ILP) and heuristic in each of the $5$ identified regions with $|\mathcal{K}| = 8$ and number of auxiliary entities (or modifications) varied as $1,3,5,7$}\label{fig:plot}
	\vspace{-12pt}
\end{figure*}

The heuristic is provided in Algorithm \ref{alg:alg2}. At any given iteration the auxiliary entity is placed at the IDR which protects the most number of entities from failure. In case of a tie the entity having highest ACFMHV value is chosen. At any given iteration the algorithm greedily maximize the number of entities protected from \emph{induced failure}. Algorithm \ref{alg:alg2} runs in polynomial time, more specifically the time complexity is $\mathcal{O}(S n (n+m)^2)$ (where $n=|A|+|B|$ and $m=$ Number of minterms in $\mathcal{F}(A,B)$).

\section{Experimental Results}
\label{ExpRes}
The solution of the heuristic is compared with the ILP to judge its efficacy. We perform the experiments on real world data sets with the IDRs generated artificially based on some predefined rules. Platts (www.platss.com) and GeoTel (www.geo-tel.com) provided the power and communication network data respectively. The power network data consisted of two types of entity --- 70 power plants and 470 transmission lines. There are three types of entity in the communication network data --- 2,690 cell towers, 7,100 fiber-lit buildings and 42,723 fiber links. The data corresponds to the Maricopa county region of Arizona, USA. To perform the experimental analysis we picked four non overlapping regions in Maricopa county. They are labelled as Region 1, 2, 3, and 4 respectively. It is to be noted that the union of these regions does not cover the entry county. For each region we filtered out the entities from our dataset and constructed the IDRs based on rules defined in \cite{sen2014identification}. 

The cardinality of $\mathcal{K}$ most vulnerable nodes was set to $8$ and was calculated using the ILP described in \cite{sen2014identification}. The number nodes failed in each region due to initial failure of the most vulnerable nodes are $28, 23, 28, 28$ respectively. We vary the number of auxiliary entities placed (or modifications) from $1$ to $7$ in steps of $2$. For each region and modification budget the number of entities protected from failure for the heuristic was compared with the ILP solution and is plotted in Figure \ref{fig:plot}. The maximum possible percentage difference of the heuristic from optimal for any region and modification budget pair is observed to be a $11.76 \%$ in Region 3 with $5$ auxiliary entities (Figure \ref{fig:plot3}). On an average the heuristic performed very near to the optimal with a difference of $6.75 \%$.

\section{Conclusion}
\label{Conc}
In this paper we introduce the auxiliary entity allocation problem in multilayer interdependent  network  using  the  IIM  model.  Entities  in  multilayer  network can be protected from an initial failure event when auxiliary entities are used to modify the IDRs. With a budget on the number of modifications, the problem is proved to be NP-complete. We provide an optimal solution using ILP and polynomial heuristic for a restricted case of the problem. The optimal solution was compared with the heuristic on real world data sets and on an average deviates $6.75\%$ from the optimal.
\begin{footnotesize}
\bibliographystyle{splncs03}
\bibliography{references,referencesBibToAdd}

\begin{thebibliography}{10}
\providecommand{\url}[1]{\texttt{#1}}
\providecommand{\urlprefix}{URL }

\bibitem{andersson2005causes}
Andersson, G., Donalek, P., Farmer, R., Hatziargyriou, N., Kamwa, I., Kundur,
  P., Martins, N., Paserba, J., Pourbeik, P., Sanchez-Gasca, J., et~al.: Causes
  of the 2003 major grid blackouts in north america and europe, and recommended
  means to improve system dynamic performance. Power Systems, IEEE Transactions
  on  20(4),  1922--1928 (2005)

\bibitem{BanHardening15}
Banerjee, J., Das, A., Zhou, C., Mazumder, A., Sen, A.: On the entity hardening
  problem in multi-layered interdependent networks. In: Computer Communications
  Workshops (INFOCOM WKSHPS), 2015 IEEE Conference on. pp. 648--653. IEEE
  (2015)

\bibitem{Zus11}
Bernstein, A., Bienstock, D., Hay, D., Uzunoglu, M., Zussman, G.: Power grid
  vulnerability to geographically correlated failures-analysis and control
  implications. arXiv preprint arXiv:1206.1099  (2012)

\bibitem{Bul10}
Buldyrev, S.V., Parshani, R., Paul, G., Stanley, H.E., Havlin, S.: Catastrophic
  cascade of failures in interdependent networks. Nature  464(7291),
  1025--1028 (2010)

\bibitem{Gao11}
Gao, J., Buldyrev, S.V., Stanley, H.E., Havlin, S.: Networks formed from
  interdependent networks. Nature Physics  8(1),  40--48 (2011)

\bibitem{Ngu13}
Nguyen, D.T., Shen, Y., Thai, M.T.: Detecting critical nodes in interdependent
  power networks for vulnerability assessment  (2013)

\bibitem{Par13}
Parandehgheibi, M., Modiano, E.: Robustness of interdependent networks: The
  case of communication networks and the power grid. arXiv preprint
  arXiv:1304.0356  (2013)

\bibitem{Ros08}
Rosato, V., Issacharoff, L., Tiriticco, F., Meloni, S., Porcellinis, S.,
  Setola, R.: Modelling interdependent infrastructures using interacting
  dynamical models. International Journal of Critical Infrastructures  4(1),
  63--79 (2008)

\bibitem{sen2014identification}
Sen, A., Mazumder, A., Banerjee, J., Das, A., Compton, R.: Identification of k
  most vulnerable nodes in multi-layered network using a new model of
  interdependency. In: Computer Communications Workshops (INFOCOM WKSHPS), 2014
  IEEE Conference on. pp. 831--836. IEEE (2014)

\bibitem{Sha11}
Shao, J., Buldyrev, S.V., Havlin, S., Stanley, H.E.: Cascade of failures in
  coupled network systems with multiple support-dependence relations. Physical
  Review E  83(3),  036116 (2011)

\bibitem{tang2012analysis}
Tang, Y., Bu, G., Yi, J.: Analysis and lessons of the blackout in indian power
  grid on july 30 and 31, 2012. In: Zhongguo Dianji Gongcheng
  Xuebao(Proceedings of the Chinese Society of Electrical Engineering).
  vol.~32, pp. 167--174. Chinese Society for Electrical Engineering (2012)

\bibitem{Zha05}
Zhang, P., Peeta, S., Friesz, T.: Dynamic game theoretic model of multi-layer
  infrastructure networks. Networks and Spatial Economics  5(2),  147--178
  (2005)

\end{thebibliography}
\end{footnotesize}
\end{document}